\documentclass[10pt,conference]{IEEEtran}

\def\real{{\mathchoice%
{\hbox{\rm\setbox1=\hbox{I}\copy1\kern-.45\wd1 R}}
{\hbox{\rm\setbox1=\hbox{I}\copy1\kern-.45\wd1 R}}
{\hbox{\scriptsize\rm\setbox1=\hbox{I}\copy1\kern-.45\wd1 R}}
{\hbox{\scriptsize\rm\setbox1=\hbox{I}\copy1\kern-.45\wd1 R}}}}
\def\Zint{{\mathchoice{\setbox1=\hbox{\sf Z}\copy1\kern-.75\wd1\box1}
{\setbox1=\hbox{\sf Z}\copy1\kern-.75\wd1\box1}
{\setbox1=\hbox{\scriptsize\sf Z}\copy1\kern-.75\wd1\box1}
{\setbox1=\hbox{\scriptsize\sf Z}\copy1\kern-.75\wd1\box1}}}
\newcommand{\complex}{ \hbox{\rm C\kern-0.45em\rule[.07em]{.02em}{.58em}%
\kern 0.43em}}

\newcommand{\be}{\begin{equation}}
\newcommand{\ee}{\end{equation}}
\newcommand{\beqr}{\begin{eqnarray}}
\newcommand{\eeqr}{\end{eqnarray}}
\newcommand{\beqrx}{\begin{eqnarray*}}
\newcommand{\eeqrx}{\end{eqnarray*}}
\newcommand{\ba}{\left[ \begin{array}}
\newcommand{\ea}{\\ \end{array} \right]}
\newcommand{\bi}{\begin{itemize}}
\newcommand{\ei}{\end{itemize}}

\newtheorem{lemma}{Lemma}
\newtheorem{theorem}{Theorem}

\usepackage{color}
\definecolor{DarkBlue}{rgb}{0.1,0.1,0.5}
\definecolor{Red}{rgb}{0.9,0.0,0.1}

\usepackage{epsfig}
\usepackage{amsmath}
\usepackage{amssymb}
\usepackage{psfrag}
\usepackage{graphicx,subfig}

\newcommand{\Prob}{\ensuremath{\mathbb{P}}}

\def\w{{\bf w}}

\def\x{{\bf x}}
\def\y{{\bf y}}

\def\e{{\bf e}}

\def\x{{\mathbf x}}

\def\w{{\bf w}}

\def\x{{\bf x}}
\def\y{{\bf y}}

\def\Prob{{\rm P}\,}

\def\real{{\rm Re}\,}

\def\be{\begin{equation}}
\def\ee{\end{equation}}
\def\ba{\left[\begin{array}}
\def\ea{\end{array}\right]}

\begin{document}

\title{Weighted $\ell_1$ Minimization for Sparse Recovery with Prior Information}

\author{
M.~Amin Khajehnejad\\
Caltech EE\\
Pasadena CA, USA \\
{\sffamily amin@caltech.edu} \and
Weiyu Xu\\
Caltech EE \\
Pasadena CA, USA \\
{\sffamily weiyu@altech.edu} \and
A.~Salman Avestimehr\\
Caltech CMI \\
Pasadena CA, USA \\
{\sffamily avestim@caltech.edu} \and
Babak Hassibi\\
Caltech EE \\
Pasadena CA, USA \\
{\sffamily hassibi@caltech.edu}}

\maketitle

\begin{abstract}
In this paper we study the compressed sensing problem of recovering a
sparse signal from a system of underdetermined linear equations when
we have prior information about the probability of each entry of the
unknown signal being nonzero. In particular, we focus on a model where
the entries of the unknown vector fall into two sets, each with a
different probability of being nonzero. We propose a weighted $\ell_1$
minimization recovery algorithm and analyze its performance using a
Grassman angle approach. We compute explicitly the relationship between
the system parameters (the weights, the number of measurements, the
size of the two sets, the probabilities of being non-zero) so that an
iid random Gaussian measurement matrix along with weighted $\ell_1$
minimization recovers almost all such sparse signals with overwhelming
probability as the problem dimension increases. This allows us to
compute the optimal weights. We also provide simulations to
demonstrate the advantages of the method over conventional $\ell_1$
optimization.
\end{abstract}

\section{Introduction } \label{sec:Intro}
Compressed sensing is an emerging technique of joint sampling and
compression that has been recently proposed as an alternative to
Nyquist sampling (followed by compression) for scenarios where
measurements can be costly~\cite{rice}. The whole premise is that
sparse signals (signals with many zero or negligible elements in a
known basis) can be recovered with far fewer measurements than the
ambient dimension of the signal itself. In fact, the major
breakthrough in this area has been the demonstration that $\ell_1$
minimization can efficiently recover a sufficiently sparse vector from
a system of underdetermined linear equations \cite{CT}.

The conventional approach to compressed sensing assumes no prior
information on the unknown signal other than the fact that it is
sufficiently sparse in a particular basis. In many applications,
however, additional prior information is available. In fact, in many
cases the signal recovery problem (which compressed sensing attempts
to address) is a detection or estimation problem in some statistical
setting. Some recent work along these lines can be found in
\cite{Baraniuk Detection} (which considers compressed detection and
estimation) and \cite{Bayesian CS} (on Bayesian
compressed sensing). In other cases, compressed sensing may be the
inner loop of a larger estimation problem that feeds prior information
on the sparse signal (e.g., its sparsity pattern) to the compressed
sensing algorithm.

In this paper we will consider a particular model for the sparse
signal that assigns a probability of being zero or nonzero to each
entry of the unknown vector. The standard compressed sensing model is
therefore a special case where these probabilities are all equal (for
example, for a $k$-sparse vector the probabilities will all be
$\frac{k}{n}$, where $n$ is the number of entries of the unknown
vector). As mentioned above, there are many situations where such
prior information may be available, such as in natural images, medical
imaging, or in DNA microarrays where the signal is often {\em block
  sparse}, i.e., the signal is more likely to be nonzero in certain
blocks rather than in others \cite{Mihailo BS-CS}.

While it is possible (albeit cumbersome) to study this model in
full generality, in this paper we will focus on the case where the
entries of the unknown signal fall into two categories: in the first
set (with cardinality $n_1$) the probability of being nonzero is
$P_1$, and in the second set (with cardinality $n_2 = n-n_1$) this
probability is $P_2$. (Clearly, in this case the sparsity will with
high probability be around $n_1P_1+n_2P_2$.) This model is rich enough
to capture many of the salient features regarding prior information,
while being simple enough to allow a very thorough analysis. While it is
in principle possible to extend our techniques to models with more
than two categories of entries, the analysis becomes increasingly
tedious and so is beyond the scope of this short paper.

The contributions of the paper are the following. We propose a
weighted $\ell_1$ minimization approach for sparse recovery where the
$\ell_1$ norms of each set are given different weights
$w_i$ ($i=1,2$). Clearly, one would want to give a larger weight to
those entries whose probability of being nonzero is less (thus further
forcing them to be zero).\footnote{A somewhat related method that uses
  weighted $\ell_1$ optimization is Candes et al \cite{Candes
    Reweighted}. The main difference is that there is no prior
  information and at each step the $\ell_1$ optimization is re-weighted
  using the estimates of the signal obtained in the last minimization
  step.} The second contribution is to compute explicitly the
relationship between the $p_i$, the $w_i$, the $\frac{n_i}{n}$,
$i=1,2$ and the number of measurements
so that the unknown signal can be recovered with overwhelming
probability as $n\rightarrow\infty$ (the so-called weak
threshold) for measurement matrices drawn from an iid Gaussian
ensemble. The analysis uses the high-dimensional geometry techniques
first introduced by Donoho and Tanner \cite{DT,D} (e.g., Grassman
angles) to obtain sharp thresholds for compressed sensing. However,
rather than use the {\em neighborliness} condition used in
\cite{DT,D}, we find it more convenient to use the null space
characterization of Xu and Hassibi \cite{Weiyu GM,StXuHa08}. The
resulting Grassmanian manifold approach is a general framework for
incorporating additional factors into compressed sensing: in
\cite{Weiyu GM} it was used to incorporate measurement noise; here it
is used to incorporate prior information and weighted $\ell_1$
optimization. Our analytic results allow us to compute the optimal
weights for any $p_1$, $p_2$, $n_1$, $n_2$. We also provide simulation
results to show the advantages of the weighted method over standard $\ell_1$
minimization.

\section{Model}
\label{sec:model}

The signal is represented by a $n\times 1$ vector $\x=(x_1,x_2,...,x_n)^T$ of real valued numbers, and is \emph{non-uniformly sparse} with sparsity factor $P_1$ over the (index) set $K_1\subset\{1,2,..n\}$ and sparsity factor $P_2$ over the set $K_2=\{1,2,...,n\} \setminus K_1$. By this, we mean that if $i\in K_1$, $x_i$ is a nonzero element with probability $P_1$ and zero with probability $1-P_1$. However, if $i\in K_2$ the probability of $x_i$ being nonzero is $P_2$. We assume that $|K_1|=n_1$ and $|K_2|=n_2=n-n_1$. The measurement matrix $\bf A$ is a $m\times n$ ($\frac{m}{n}=\delta<1$) matrix with i.i.d $\mathcal{N}(0,1)$ entries. The observation vector is denoted by $\y$ and obeys the following:
\begin{equation}
\y = \bf A\x
\end{equation}

As mentioned in Section \ref{sec:Intro}, $\ell_1$-minimization can recover a vector $\x$ with $k=\mu n$ non-zeros, provided $\mu$ is less than a known function of $\delta$. $\ell_1$ minimization has the following form:

\begin{equation}\label{eq:l_1}
\min_{\bf A\x=\y}{\|\x\|_1}
\end{equation}

(\ref{eq:l_1}) is a linear programming and can be solved polynomially fast ($O(n^3)$). However, it fails to encapsulate additional prior information of the signal nature, might there be any such information. One might simply think of modifying (\ref{eq:l_1}) to a weighted $\ell_1$ minimization as follows:

\begin{equation} \label{eq:weighted l_1}
\min_{\bf A\x=\y}{\|\x\|_{\w1}}=\min_{Ax=y}{\sum_{i=1}^{n}{w_i |x_i|}}
\end{equation}

The index $\w$ is an indication of the $n\times 1$ positive weight vector. Now the question is what is the optimal set of weights, and can one improve the recovery threshold using the weighted $\ell_1$ minimization of (\ref{eq:weighted l_1}) with those weights rather than (\ref{eq:l_1})? We have to be more clear with the objective at this point and what we mean by extending the recovery threshold. First of all note that the vectors generated based on the model described above can have any arbitrary number of nonzeros. However, their support size is typically (with probability arbitrary close to one) around $n_1P_1+n_2P_2)$. Therefore, there is no such notion of strong threshold as in the case of~\cite{DT}. We are asking the question of for what $P_1$ and $P_2$ signals generated based on this model can be recovered with overwhelming probability as $n\rightarrow \infty$. Moreover we are wondering if by adjusting $w_i$'s according to $P_1$ and $P_2$ can one extend the typical sparsity to dimension ratio ($\frac{n_1P_1+n_2P_2}{n}$) for which reconstruction is successful with high probability. This is the topic of next section.

\section{Computation of the Weak Threshold}

Because of the partial symmetry of the sparsity of the signal we know that the optimum weights  should take only two positive values $W_1$ and $W_2$. In other words\footnote{Also we may assume WLG that $W_1=1$}
\begin{equation*}
\forall i\in\{1,2,\cdots,n\} ~~~w_i=\left\{\begin{array}{c}W_1 ~ \text{if $i\in K_1$}\\ W_2 ~ \text{if $i\in K_2$}
\end{array}\right.
\end{equation*}

Let $\x$ be a random sparse signal generated based on the non-uniformly sparse model of section \ref{sec:model} and be supported on the set $K$. $K$ is called \emph{$\epsilon$-typical} if $\left||K\cap K_1| -n_1P_1\right|\leq \epsilon n$ and $\left||K\cap K_2|-n_2P_2\right|\leq \epsilon n$. Let $E$ be the event that $x$ is recovered by (\ref{eq:weighted l_1}). Then:
\begin{eqnarray*}
\Prob[E^c] &=& \Prob[E^c | \text{$K$ is $\epsilon$-typical} ]\Prob[\text{$K$ is $\epsilon$-typical}]  \\
&+& \Prob[E^c|  \text{$K$ not $\epsilon$-typical} ]\Prob[\text{$K$ not $\epsilon$-typical}]
\end{eqnarray*}

\noindent For any fixed $\epsilon > 0 $  $\Prob[\text{$K$ not $\epsilon$-typical}]$ will exponentially approach zero as $n$ grows according to the law of large numbers. So, to bound the probability of failed recovery we may assume that $K$ is $\epsilon$-typical for any small enough $\epsilon$. Therefore we just consider the case $|K|=k=n_1P_1+n_2P_2$. Similar to the null-space condition of \cite{StXuHa08}, we present a  necessary and sufficient condition for $\x$ to be the solution to (\ref{eq:weighted l_1}). It is as follows:

\begin{equation*}
\forall Z\in \mathcal{N}(A) ~~~ \sum_{i\in K}{w_i|Z_i|} \leq \sum_{i\in \overline{K}}{w_i|Z_i|}
\end{equation*}

Where $\mathcal{N}(A)$ denotes the right nullspace of $A$. We can upper bound $\Prob(E^c)$ with $P_{K,-}$ which is the probability that a vector $\x$ of a specific sign pattern (say non-positive) and supported on the specific set $K$ is not recovered correctly by (\ref{eq:weighted l_1}) (A difference between this upper bound and the one in~\cite{Weiyu GM} is that here there is no ${n \choose k }2^k$ factor, and that is because we have fixed the support set $K$ and the sign pattern of $\x$). Exactly as done in \cite{Weiyu GM}, by restricting  $\x$ to the cross-polytope $\{\x\in \mathbb{R}^n\mid \|x\|_{\w1}=1\}$\footnote{This is because the restricted polytope totally surrounds the origin in $\mathbb{R}^n$ }, and noting that $\x$ is on a $(k-1)$-dimensional face $F$ of the skewed cross-polytope $\text{SP}=\{\y\in R^n~|~\|\y\|_{\w1}\leq 1\}$, $P_{K,-}$ is essentially the probability that a uniformly chosen $(n-m)$-dimensional subspace $\Psi$ shifted by the point $x$, namely $(\Psi +
x)$, intersects SP nontrivially at some other point besides $\x$. $P_{K,-}$ is then interpreted as the complementary Grassmann angle~\cite{Grunbaumpaper} for
the face F with respect to the polytope SP under the Grassmann manifold $Gr_(n-m)(n)$. Building on the works by L.A.Santal\"{o} \cite{santalo} and
P.McMullen \cite{McMullen} etc. in high dimensional integral geometry and convex polytopes, the complementary Grassmann angle for
the $(k-1)$-dimensional face $F$ can be explicitly expressed as the sum of products of internal angles and external angles \cite{Grunbaumbook}:

\begin{equation}
2\times \sum_{s \geq 0}\sum_{G \in \Im_{m+1+2s}(\text{SP})}
{\beta(F,G)\gamma(G,\text{SP})}, \label{eq:angformula}
\end{equation}
where $s$ is any nonnegative integer, $G$ is any $(m+1+2s)$-dimensional face of the skewed crosspolytope ($\Im_{m+1+2s}(\text{SP})$ is the set of all such faces), $\beta(\cdot,\cdot)$ stands for the internal angle and $\gamma(\cdot,\cdot)$ stands for the external angle. The internal angles and external angles are basically defined as follows \cite{Grunbaumbook}\cite{McMullen}:
\begin{itemize}
\item An internal angle $\beta(F_1, F_2)$ is the fraction of the
hypersphere $S$ covered by the cone obtained by observing the face
$F_2$ from the face $F_1$. The internal angle $\beta(F_1, F_2)$ is
defined to be zero when $F_1 \nsubseteq F_2$ and is defined to be
one if $F_1=F_2$.
\item An external angle $\gamma(F_3, F_4)$ is the fraction of the
hypersphere $S$ covered by the cone of outward normals to the
hyperplanes supporting the face $F_4$ at the face $F_3$.
The external angle $\gamma(F_3, F_4)$ is defined to be zero when $F_3
\nsubseteq F_4$ and is defined to be one if $F_3=F_4$.
\end{itemize}

Note that $F$ here is a typical face of SP corresponding to a typical set $K$. $\beta(F,G)$ depends not only on the dimension of the face $G$, but also depends on the number of its vertices supported on $K_1$ and $K_2$. In other words if $G$ is supported on a set $L$, then $\beta(F,G)$ is only a function of $|L\cap K_1|$ and $|L\cap K_2|$. So we write $\beta(F,G)=\beta(t_1,t_2)$ and similarly $\gamma(G,SP)=\gamma(t_1,t_2)$ where $t_1=|L\cap K_1|-n_1P_1$ and $t_2=|L\cap K_2|-n_2P_2$. Combining the notations and counting the number of faces $G$, (\ref{eq:angformula}) leads to:

\begin{align}
&\Prob(E^c) \leq  \nonumber \\
&\sum_{{\tiny\begin{array}{c}0\leq t_1 \leq(1-P_1)n_1\\0\leq t_2 \leq (1-P_2)n_2\\ t_1+t_2 > m-k+1\end{array}}} 2^{t_1+t_2}{{\tiny(1-P_1)}n_1\choose t_1}{(1-P_2)n_2\choose t_2}\times \nonumber \\
&\beta(t_1,t_2)\gamma(t_1,t_2)~+~O(e^{-cn}) \label{eq:sumformula}
\end{align} 
\noindent for some $c>0$. As $n\rightarrow\infty$ each term in (\ref{eq:sumformula}) behaves like $exp\{n\psi_{com}(t_1,t_2)-n\psi_{int}(t_1,t_2)-n\psi_{ext}(t_1,t_2)\}$ where $\psi_{com}$ $\psi_{int}$ and $\psi_{ext}$ are the  combinatorial exponent, the internal angle exponent and the external angle exponent of the each term respectively. It can be shown that the necessary and sufficient condition for (\ref{eq:sumformula}) to tend to zero is that $\psi(t_1,t_2)= \psi_{com}(t_1,t_2)-\psi_{int}(t_1,t_2)- \psi_{ext}(t_1,t_2)$ be uniformly negative for all $t_1$ and $t_2$ in (\ref{eq:sumformula}).

In the following sub-sections we will try to evaluate the internal and external angles for a typical face $F$, and a face $G$ containing $F$and try to give closed form upper bounds for them. We combine the terms together and compute the exponents using Laplace method in section \ref{sec:exp cal.} and derive thresholds for nonnegativity of the cumulative exponent using.

\subsection{Derivation of the Internal Angles}
\label{sec: Derivation of Inter.}
Suppose that $F$ is a \emph{typical} $(k-1)$-dimensional face of the skewed
cross-polytope
\begin{equation*}
\text{SP}=\{\y\in R^n~|~ \|\y\|_{\w1}=\sum_{i=1}^{n}w_i|\y_i|
\leq 1\}
\end{equation*}
supported on the subset $K$ with $|K|=k\approx n_1P_1 + n_2P_2$. Let $G$ be a $l-1$ dimensional face of SP supported on the set $L$ with $F\subset G$. Also, let $|L\cap K_1|=t_1$ and $|L\cap K_2|=t_2$.

First we can prove the following lemma:
\begin{lemma}
Let $\text{Con}_{F^{\perp},G}$ be the positive cone of all the vectors $\x\in \mathbb{R}^{n}$ that take the form:
\begin{equation}
-\sum_{i=1}^{k}{b_i \times e_i}+\sum_{i=k+1}^{l}{b_i \times e_i},
\label{eq:vform}
\end{equation}
where $b_i, 1 \leq i \leq l$ are nonnegative real numbers  and
\begin{eqnarray*}
\sum_{i=1}^{k}{w_i b_i}=\sum_{i=k+1}^{l}{w_i b_{i}}   ~~~ \frac{b_1}{w_1}=\frac{b_2}{w_2}=\cdots=\frac{b_k}{w_k}
\end{eqnarray*}

Then
\begin{eqnarray}
&&\nonumber\int_{\text{Con}_{F^{\perp},G}}{e^{-\|\x\|^2}}\,d\x=\beta(F,G)
V_{l-k-1}(S^{l-k-1}) \\
&&\times \int_{0}^{\infty}{e^{-r^2}} r^{l-k-1}\,dx =\beta(F,G) \cdot
\pi^{(l-k)/2},
 \label{eq:inaxchdirect}
\end{eqnarray}
where $V_{l-k-1}(S^{l-k-1})$ is the spherical volume of the
$(l-k-1)$-dimensional sphere $S^{l-k-1}$.  
\end{lemma}
\begin{proof}
Omitted for brevity
\end{proof}

From (\ref{eq:inaxchdirect}) we can find the expression for the internal angle.
Define $U\subseteq
\mathbb{R}^{l-k+1}$ as the set of all nonnegative vectors $(x_1,x_2,\cdots,x_{l-k+1} )$ satisfying:
\begin{center}
{\scriptsize $x_p \geq 0,~1 \leq p\leq l-k+1$}  $(\sum_{p=1}^{k}w^2_p)x_1 = \sum_{p=k+1}^{l}w^2_p x_{p-k+1}$
\end{center}
and define $f(x_1,~\cdots,~ x_{l-k+1}):U \rightarrow
\text{Con}_{F^{\perp},G}$ to be the linear and bijective map
\begin{align*}
f(x_1,\cdots,x_{l-k+1})=-\sum_{p=1}^{k} x_1w_p
\e_p+\sum_{p=k+1}^{l} x_{p-k+1}w_p \e_p
\end{align*}
Then

\begin{align}
&\int_{ \text{Con}_{F^{\perp},G} }{e^{-\|{\x}'\|^2}}\,d{\x}' = \int_{U}{e^{-\|f(\x)\|^2}}\,df(\x) \nonumber \\
&= |J(A)| \int_{\Gamma}{e^{-\|f(x)\|^2}}\,d{x_2} \cdots dx_{l-k+1} \nonumber \\
&=|J(A)| \int_{\Gamma}e^{-(\sum_{p=1}^{k}{w_p^2})x_1^2-\sum_{p=k+1}^{l}{w_p^2x_{p-k+1}^2} } \,d{x_2} \cdots
dx_{l-k+1} \label{eq:vsigular}
\end{align}

$\Gamma$ is the region described by 
\begin{equation}
(\sum_{p=1}^{k}w^2_p)x_1 = \sum_{p=k+1}^{l}w^2_p x_{p-k+1} , x_p \geq 0~2 \leq p \leq l-k+1
\end{equation}
where $|J(A)|$ is due to the change of
integral variables and is essentially the determinant of the Jacobian of the variable transform given by the $l\times l-k$ matrix $A$ given by:  

\begin{align}
A_{i,j}= \left\{\begin{array}{cc} -\frac{1}{\Omega}w_iw_{k+j}^2 &{\scriptsize 1\leq i\leq k ,1\leq j\leq l-k} \\ w_{i} & {\scriptsize k+1\leq i\leq l , j=i-k} \\ 0 & \text{Otherwise} \end{array}\right.
\end{align}

\noindent where $\Omega = \sum_{p=1}^{k}w_p^2$. Now $|J(A)|=\sqrt{\det(A^T A)}$. By finding the eigenvalues of $A^T A$ we obtain:
\begin{equation}
|J(A)| = W_1^{t_1}W_2^{t_2}\sqrt{\frac{\Omega + t_1W_1^2 + t_2W_2^2}{\Omega}}
\end{equation}

\noindent Now we define a random variable
\begin{equation*}
Z = (\sum_{p=1}^{k}w^2_p)X_1 - \sum_{p=k+1}^{l}w^2_p X_{p-k+1}
\end{equation*}
where $X_1, X_2, \cdots, X_{l-k+1}$ are independent random variables,
with $X_p \sim HN(0,\frac{1}{2w_{p+k-1}^2})$, $2 \leq p \leq (l-k+1)$, as
half-normal distributed random variables and $X_1\sim N(0,
\frac{1}{2\sum_{p=1}^{k}{w_p^2}})$ as a normal distributed random variable. Then by
inspection, (\ref{eq:vsigular}) is equal to
\begin{equation*}
  C p_{Z}(0).
\end{equation*}
where $p_{Z}(\cdot)$ is the probability density function for the
random variable $Z$ and $p_{Z}(0)$ is
the probability density function $p_{Z}(\cdot)$ evaluated at the
point $Z=0$, and
{\small
\begin{align}
C&=\frac{\sqrt{\pi}^{l-k+1}}{2^{l-k}}\prod_{q=k+1}^{l}\frac{1}{w_q}{\sqrt{\sum_{p=1}^{k}w_p^2}}~|J(A)| \nonumber\\
&= \frac{\sqrt{\pi}^{l-k+1}}{2^{l-k}} \sqrt{(n_1P_1+t_1)W_1^2+(n_2P_2+t_2)W_2^2} \label{eq:C}
\end{align}
}

Combining (\ref{eq:inaxchdirect}) and (\ref{eq:vsigular}):
\begin{equation}\label{eq:internal angle formula}
\beta(t_1,t_2) = \pi^{\frac{k-l}{2}}Cp_Z(0)
\end{equation}

%
%
%
%
%
%

\begin{figure}[t]
\centering
  \includegraphics[width=0.35\textwidth]{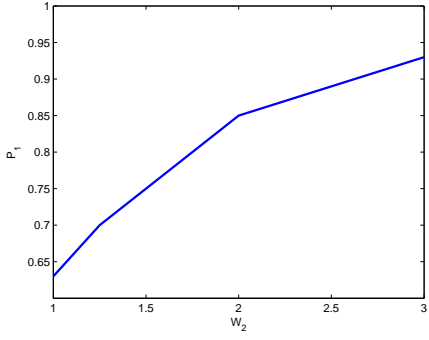}\\
\caption{{\scriptsize  Recoverable $P_1$ threshold as a function of $W_2$. $P_2=0.1$, $m=0.75n$}}
\label{fig:exponent}
\end{figure}

\subsection{Derivation of the External Angle}
Without loss of generality, assume $K=\{n-k+1, \cdots,n\}$. Consider
the $(l-1)$-dimensional face
\begin{equation*}
G=\text{conv}\{\frac{\e_{n-l+1}}{w_{n-l+1}} , ... ,\frac{\e_{n-k}}{w_{n-k}} , \frac{\e_{n-k+1}}{w_{n-k+1}},
...,\frac{\e_{n}}{w_{n}} \}
\end{equation*}
of the skewed cross-polytope $\text{SP}$. The $2^{n-l}$ outward
normal vectors of the supporting hyperplanes of the facets
containing $G$ are given by
\begin{equation*}
\{\sum_{i=1}^{n-l} j_{i}w_i \e_i+\sum_{p=n-l+1}^{n} w_i \e_i, j_{i}\in\{-1,1\}\}.
\end{equation*}

\begin{figure*}[t]
  \centering
  \subfloat[]{\label{fig:PRvsP1firstsetting}\includegraphics[width=0.35\textwidth]{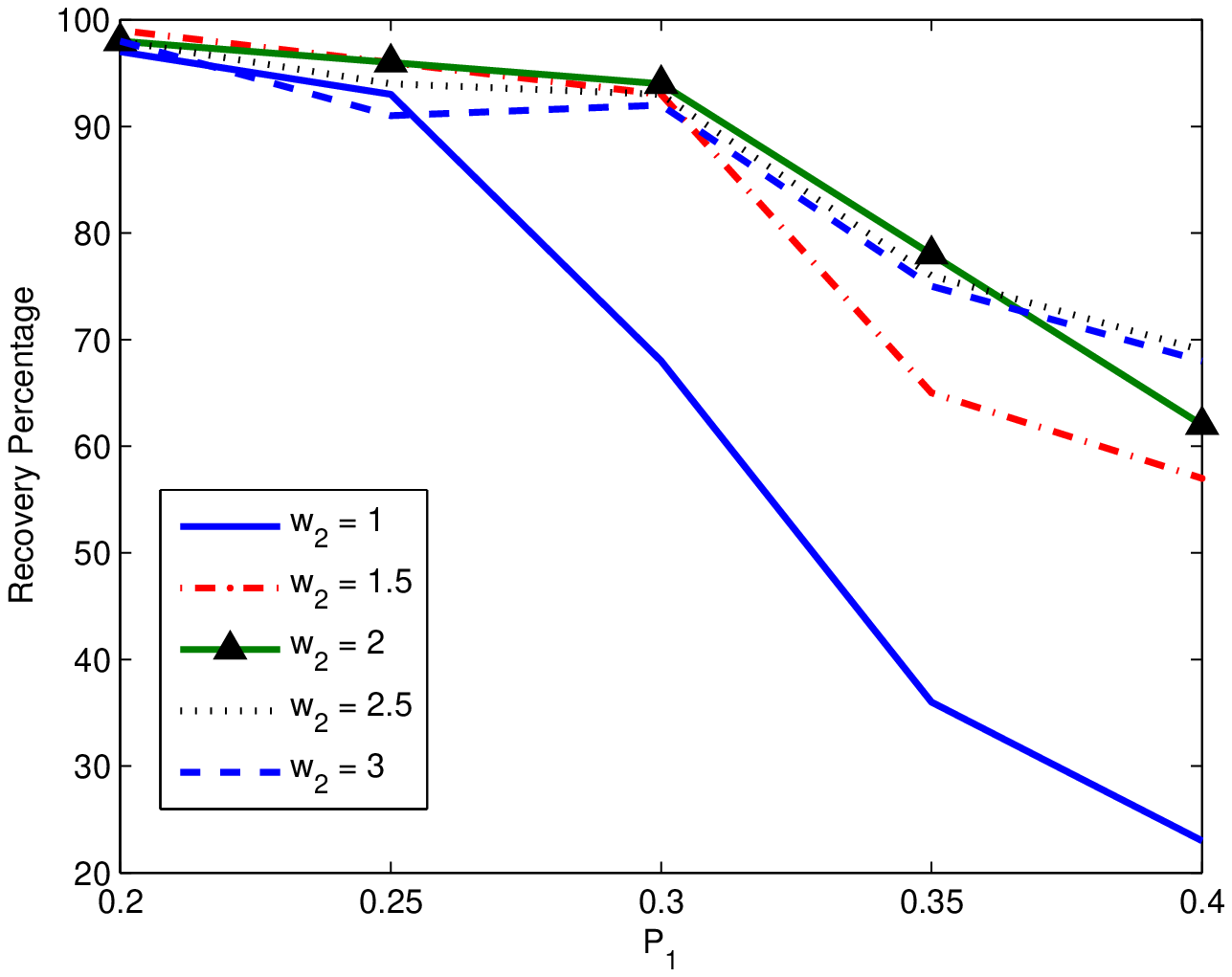}}
  \hspace{1 in}
  \subfloat[]{\label{fig:PRvsP1Optimalw}\includegraphics[width=0.35\textwidth]{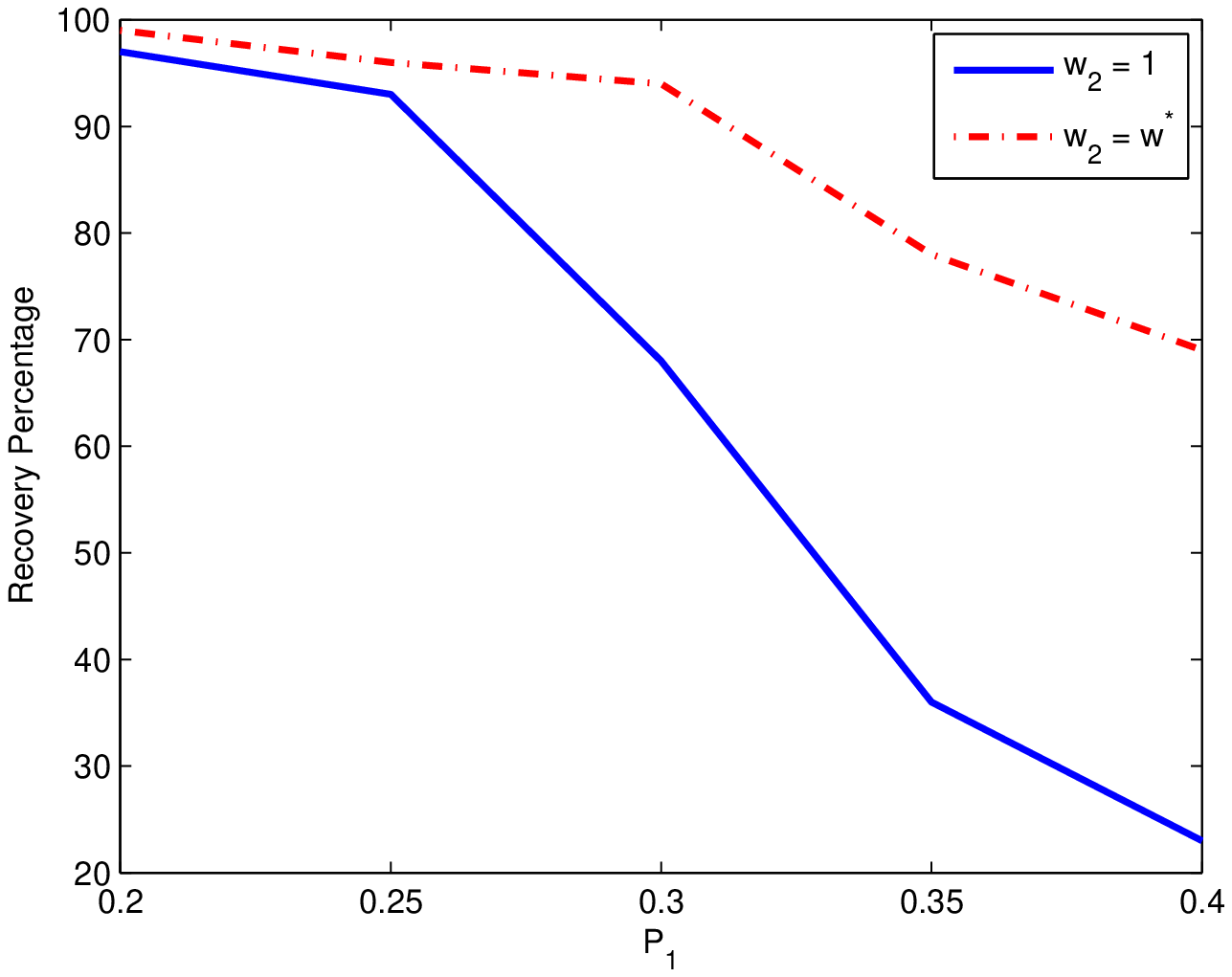}}
    \caption{{\scriptsize Successful recovery percentage for weighted $\ell_1$ minimization with different weights and suboptimal weights in a nonuniform sparse setting. $P_2=0.05$ and $m=0.5n$}}
  \label{fig:weightedimprovement}
\end{figure*}

Then the outward normal cone $c(G, \text{SP})$ at the face $G$ is
the positive hull of these normal vectors. Thus
\begin{align}
\int_{c(G,\text{SP})}{e^{-\|x\|^2}}\,dx&=\gamma(G,SP)
V_{n-l}(S^{n-l}) \int_{0}^{\infty}{e^{-r^2}}r^{n-l}\,dx \nonumber\\
&=\gamma(G,\text{SP}).\pi^{(n-l+1)/2},
 \label{eq:axch}
\end{align}
where $V_{n-l}(S^{n-l})$ is the spherical volume of the
$(n-l)$-dimensional sphere $S^{n-l}$. 
Now define $U$ to be the set
\begin{equation*}
\{x \in R^{n-l+1} \mid x_{n-l+1} \geq 0, | x_i/w_i| \leq
x_{n-l+1}, 1\leq i \leq (n-l)\}
\end{equation*}
and define $f(x_1,~\cdots,~ x_{n-l+1}):U \rightarrow c(G,\text{SP})$
to be the linear and bijective map
\begin{eqnarray*}
f(x_1,~\cdots,~ x_{n-l+1})&=&\sum_{i=1}^{n-l} x_i
\e_i+\sum_{i=n-l+1}^{n} w_i x_{n-l+1}\e_i .
\end{eqnarray*}
Then
\small{
\begin{align}
&\int_{c(G,\text{SP})}{e^{-\|x'\|^2}}\,dx' \nonumber = |J(A)|\int_{U}{e^{-\|f(x)\|^2}}\,dx \nonumber\\
&=|J(A)|\int_{0}^{\infty}
\int_{-w_1x_{n-l+1}}^{w_1x_{n-l+1}}
\cdots\int_{-w_{n-l}x_{n-l+1}}^{w_{n-l}x_{n-l+1}}\nonumber\\
& e^{-x_1^2-\cdots
-x_{n-l}^2-(\sum_{i=n-l+1}^{n}w_i^2)x_{n-l+1}^{2} } \,dx_{1} \cdots \,dx_{n-l+1}\nonumber\\
&=|J(A)|\int_{0}^{\infty}
e^{-(\sum_{i=n-l+1}^{n}w_i^2)x^2}\times\nonumber\\
& \left(\int_{-W_1 x}^{W_1 x} e^{-y^2}\,dy \right )^{(1-P_1)n_1-t_1} \left(\int_{-W_2x}^{W_2x} e^{-y^2}\,dy \right )^{(1-P_2)n_2-t_2} \,dx \nonumber\\
&={\tiny 2^{n-l}\int_{0}^{\infty}e^{-x^2} \left (\int_{0}^{ \frac{W_1x}{\ \xi} } e^{-y^2}\,dy \right)^{r_1} \left (\int_{0}^{ \frac{W_2 x}{\xi} } e^{-y^2}\,dy \right)^{r_2}\,dx,} \nonumber\\
\label{eq:external angle formula}
\end{align}
}
where
$\xi=\xi(t_1,t_2)=\sqrt{\sum_{i=n-l+1}^{n}w_i^2},~ r_1=(1-P_1)n_1-t_1~~r_2=(1-P_2)n_2-t_2$. \normalsize $|J(A)|=\sqrt{\det(A^TA)}=\xi$ is resulting from the change of variable in the integral.

\begin{figure}[t]
\centering
  \includegraphics[width=0.35\textwidth]{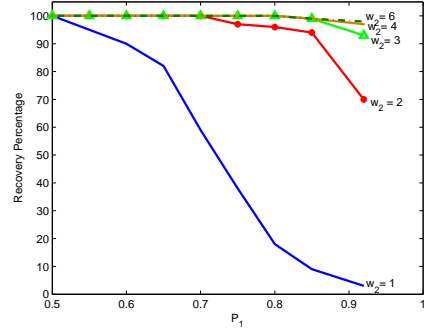}\\
\caption{{\scriptsize Successful recovery percentage for different weights. $P_2=0.1$ and $m=0.75n$}}
\label{fig:PRvsP1}
\end{figure}
\section{Exponent Calculation}
\label{sec:exp cal.}
Using the Laplace method we compute the angle exponents. They are given in the following theorems, the proofs of which are omitted for brevity.
we assume $n_1=\gamma_1n$, $n_2=\gamma_2n$ and WLG $W_1=1$, $W_2=W$.

\begin{theorem}\label{external exponent}
Let $t_1=t_1'n$, $t_2=t_2'n$, $g(x)=\frac{2}{\sqrt{\pi}}e^{-\frac{x^2}{2}}$, $G(x)=\frac{2}{\sqrt{\pi}}\int_{0}^{x}e^{-y^2}dy$. Also define $C=(t_1'+\gamma_1P_1)+W^2(t_2'+\gamma_2P_2)$, $D_1=\gamma_1(1-P_1)-t_1'$ and $D_2=\gamma_2(1-P_2)-t_2'$. Let $x_0$ be the unique solution to $x$ of the following:
\begin{equation*}
2C-\frac{g(x)D_1}{xG(x)}-\frac{Wg(Wx)D_2}{xG(Wx)}=0
\end{equation*} 
Then 
\begin{equation}
\psi_{ext}(t_1,t_2) = Cx_0^2-D_1\log{G(x_0)}-D_2\log{G(Wx_0)}
\end{equation}

\end{theorem}  
\vspace{5pt}

\begin{theorem}\label{internal exponent}
Let $b=\frac{t_1+W^2t_2}{t_1+t_2}$ and  $\varphi(.)$ and $\Phi(.)$ be the standard Gaussian pdf and cdf functions respectively. Also let $Q(s)=\frac{t_1\varphi(s)}{(t_1+t_2)\Phi(s)}+\frac{Wt_2\varphi(Ws)}{(t_1+t_2)\Phi(Ws)}$. Define the function $\hat{M}(s)=-\frac{s}{Q(s)}$ and solve for $s$ in $\hat{M}(s)=\frac{m}{mb+\Omega}$. Let the unique solution be $s^*$ and set $y=s^*(b-\frac{1}{\hat{M}(s^*)})$. Compute the rate function $\Lambda^*(y)= sy -\frac{t_1}{t_1+t_2}\Lambda_1(s)-\frac{t_1}{t_1+t_2}\Lambda_1(Ws)$ at the point $s=s^*$, where $\Lambda_1(s) = \frac{s^2}{2} +\log(2\varphi(s))$.
The internal angle exponent is then given by:
\begin{equation}
\psi_{int}(t_1,t_2) = (\Lambda^*(y)+\frac{m}{2\Omega}y+\log2)(t_1'+t_2')
\end{equation}

\end{theorem}

As an illustration of these results, for $P_2=0.1$ and $\delta=\frac{m}{n}=0.75$ using Theorems \ref{internal exponent} and \ref{external exponent} and combining the exponents with the combinatorial exponent, we have calculated the threshold for $P_1$ for different values of $w_2$ in the range $[1,3]$  , below which the signal can be recovered. The curve is depicted in Figure \ref{fig:exponent}. As expected, the curve is suggesting that in this setting weighted $\ell_1$ minimization boosts the weak threshold in comparison with $\ell_1$ minimization. This is verified in the next section by some examples.

\section{Simulation}
We demonstrate by some examples that appropriate weights can boost the recovery percentage. We fix $P_2$ and  $n=2m=200$, and try $\ell_1$ and weighted $\ell_1$ minimization for various values of $P_1$. We choose $n_1=n_2=\frac{n}{2}$ Figure \ref{fig:PRvsP1firstsetting} shows one such comparison for $P_2=0.05$ and different values of $w_2$. Note that the optimal value of $w_2$ varies as $P_1$ changes. Figure \ref{fig:PRvsP1Optimalw} illustrates how the optimal weighted $\ell_1$ minimization surpasses the ordinary $\ell_1$ minimization. The optimal curve is basically achieved by selecting the best weight of  Figure \ref{fig:PRvsP1firstsetting} for each single value of $P_1$. Figure \ref{fig:PRvsP1} shows the result of simulations in another setting where $P_2=0.1$ and $m=0.75n$ (similar to the setting of the previous section). It is clear from the figure that the recovery success threshold for $P_1$ has been shifted higher when using weighted $\ell_1$ minimization rather than standard $\ell_1$ minimization. Note that this result very well matches the theoretical result of Figure \ref{fig:exponent}.

\bibliographystyle{IEEEbib}

\end{document}